\documentclass[runningheads]{llncs}

\usepackage[T1]{fontenc}

\usepackage{amsfonts}

\usepackage{latexsym,amsmath,amsthm,epsfig,amssymb,graphics,amscd,amsfonts}
\usepackage{float}
\usepackage{comment}
\usepackage{verbatim}
\usepackage{color}
\usepackage[usenames,dvipsnames,table]{xcolor}
\usepackage{indentfirst}
\usepackage{latexsym,amsmath,epsfig,amssymb,graphics}
\usepackage{stmaryrd}
\usepackage{mathrsfs}
\usepackage{algorithm, algorithmic}
\usepackage{amsfonts}
\usepackage{makeidx}
\usepackage{mathtools}
\usepackage{makecell}
\usepackage{multirow}
\usepackage{siunitx}
\usepackage{caption}
\usepackage{subcaption}
\usepackage{wrapfig}
\usepackage{booktabs}
\usepackage{nicefrac}
\usepackage[hyperfootnotes=true,colorlinks,linkcolor=RedViolet,anchorcolor=RedViolet,citecolor=blue,urlcolor=black]{hyperref}
\usepackage[capitalize]{cleveref}

\crefname{lemma}{Lem.}{Lem.}
\crefname{example}{Exmp.}{Exmp.}
\crefname{section}{Sect.}{Sect.}
\Crefname{appendix}{Appx.}{Appx.}
\crefname{definition}{Def.}{Def.}
\crefname{theorem}{Thm.}{Thm.}
\crefname{corollary}{Cor.}{Cor.}
\crefname{algorithm}{Alg.}{Alg.}
\crefname{proposition}{Prop.}{Prop.}

\usepackage{todonotes,xspace}
\usepackage[misc]{ifsym}
\usepackage{orcidlink}

\newcommand{\mygreen}[1]{{\color{green!60!gray} #1}}
\newcommand{\myred}[1]{{\color{red!60!gray} #1}}
\newcommand{\myyellow}[1]{{\color{yellow!60!gray} #1}}
\usepackage{bbding}
\newcommand{\mycrossmark}{\myred{\XSolidBrush}}
\newcommand{\mycheckmark}{\mygreen{\Checkmark}}
\newcommand{\myquestionmark}{\myyellow{\large \textbf{?}}}

\newcommand{\Real}{\mathbb{R}}
\newcommand{\Nat}{\mathbb{N}}

\newcommand{\defeq}{\hat{~=~}}

\newcommand{\bm}[1]{\boldsymbol{#1}}
\newcommand{\seq}[1]{\boldsymbol{#1}}
\newcommand{\norm}[1]{
    \lVert #1 \rVert
}

\newcommand{\Set}[2]{
    \left\{ #1 \mid  #2 \right\}
}

\newcommand{\tool}[1]{\textsc{#1}}
\newcommand{\domain}{\mathcal X}
\newcommand{\init}{\mathcal I}
\newcommand{\unsafe}{\mathcal U}

\newcommand{\Lie}{\mathfrak L}
\newcommand{\traj}{\xi}
\newcommand{\reach}{\mathcal R}

\newcommand{\qm}[0]{\textbf{\textrm{QM}}}

\newcommand{\homo}[1]{\tilde{#1}}
\RequirePackage{todonotes,xspace}
\newcommand{\oomit}[1]{ }

\makeatletter
\def\thanks#1{\protected@xdef\@thanks{\@thanks
        \protect\footnotetext{#1}}}
\makeatother

\begin{document}

\title{On Completeness of SDP-Based Barrier Certificate Synthesis over Unbounded Domains
% \thanks{This work has been partially funded   by the National Key R\&D Program of China under grant No.\ 2022YFA1005101, by the NSFC under grant No.\ 62192732 and 62032024, by the CAS Project for Young Scientists in Basic Research under grant No.\ YSBR-040. 
% }
}
% Generalizing SDP-Based Barrier Certificate Synthesis to Unbounded Domains by Dropping Archimedean Condition}

\titlerunning{On Completeness of SDP-Based BC-Synthesis over Unbounded Domains}
\authorrunning{H.~Wu et al.}

\author{Hao Wu\inst{1}\orcidlink{0000-0001-9368-4744} \and
Shenghua Feng\inst{2}\orcidlink{0000-0002-5352-4954}  \and
Ting Gan \inst{3}\orcidlink{0000-0002-4880-5129} \and
Jie Wang \inst{4} \orcidlink{0000-0002-9681-1451} \and
Bican Xia \inst{5}\orcidlink{0000-0002-2570-2338} \and
Naijun Zhan \inst{6,1}\inst{(}\Envelope\inst{)}\orcidlink{0000-0003-3298-3817}
}

\institute{
State Key Lab. of Computer Science, Institute of Software, University of Chinese Academy of Sciences, Beijing, China
\email{\{wuhao,znj\}@ios.ac.cn}
% \and 
% University of Chinese Academy of Sciences, Beijing, China 
\and 
Zhongguancun Laboratory, Beijing, China 
\email{fengsh@zgclab.edu.cn}
\and
School of Computer Science, Wuhan University, China 
\email{ganting@whu.edu.cn}
\and
Academy of Mathematics and Systems Science, CAS, Beijing, China
\email{wangjie212@amss.ac.cn}
\and 
School of Mathematical Sciences, Peking University, Beijing, China 
\email{xbc@math.pku.edu.cn}
\and
School of Computer Science, Peking University, Beijing, China
}

% \begin{CCSXML}
% <ccs2012>
%    <concept>
%        <concept_id>10010520.10010553</concept_id>
%        <concept_desc>Computer systems organization~Embedded and cyber-physical systems</concept_desc>
%        <concept_significance>500</concept_significance>
%        </concept>
%    <concept>
%        <concept_id>10003752.10003790.10002990</concept_id>
%        <concept_desc>Theory of computation~Logic and verification</concept_desc>
%        <concept_significance>500</concept_significance>
%        </concept>
%    <concept>
%        <concept_id>10002950.10003714.10003716.10011138.10010042</concept_id>
%        <concept_desc>Mathematics of computing~Semidefinite programming</concept_desc>
%        <concept_significance>500</concept_significance>
%        </concept>
%  </ccs2012>
% \end{CCSXML}

% \ccsdesc[500]{Computer systems organization~Embedded and cyber-physical systems}
% \ccsdesc[500]{Theory of computation~Logic and verification}
% \ccsdesc[500]{Mathematics of computing~Semidefinite programming}

\maketitle

\begin{abstract}
Barrier certificates, serving as differential invariants that witness system safety, play a crucial role in the verification of cyber-physical systems (CPS).
Prevailing computational methods for synthesizing barrier certificates are  based on semidefinite programming (SDP) by exploiting \emph{Putinar Positivstellensatz}.
Consequently, these approaches are limited by the \emph{Archimedean condition}, which requires all variables to be bounded, i.e., systems are defined over bounded domains.
For systems over unbounded domains, unfortunately, existing methods become incomplete and may fail to identify potential barrier certificates.

~~~~In this paper, we address this limitation for the unbounded cases.
We first give a complete characterization of polynomial barrier certificates by using \emph{homogenization}, a recent technique in the optimization community to reduce an unbounded optimization problem to a bounded one.
Furthermore, motivated by this formulation, we introduce the definition of \emph{homogenized systems} and propose a complete characterization of a family of non-polynomial barrier certificates with more expressive power. 
Experimental results demonstrate that our two approaches are more effective while maintaining a comparable level of efficiency.

\keywords{
Safety \and Barrier certificates \and Semidefinite programming \and Homogenization}

\end{abstract}

\section{Introduction}\label{sec:intro}

% \paragraph{Background}
With recent advancements in optimization theory and computational techniques, Cyber-Physical Systems (CPS), which involve the seamless integration of physical components and software systems, have proliferated across various application domains. 
A significant subset of CPS, known as safety-critical systems, presents a heightened level of concern. 
Failures or malfunctions in such systems can lead to severe safety risks for individuals and the environment. 
Examples of safety-critical CPS include aircraft, automobiles, integrated medical devices, nuclear power plants, and biological systems.
As a result, ensuring the safety of these systems has become a primary focus of extensive academic research. 

% \revcomment{no discussion is given to hybrid systems beyond the claim that the method can be extended to use with hybrid systems in the conclusion. Supporting evidence or an explanation of the approach to handle discrete changes to the dynamical system are needed.}
% Hybrid systems are mathematical models that involves both continuous dynamics and discrete transitions, and hence are widely used for modelling CPS.
One of the key challenges in CPS verification is the safety problem (or dually, the reachability problem), i.e., to demonstrate that a system, starting from its initial states, never enters an unsafe region.
In general, the safety problem of CPS is undecidable \cite{henzinger95stoc}.
The most challenging aspect of such problem lies in reasoning about the continuous dynamics, which are typically described by ordinary differential equations (ODEs).

% Reachability analysis aims to compute or approximate the set of reachable states.
% The choice of different set representations leads to various approaches in this field. 
% For example, one can utilize geometric objects (such as hyper-rectangles \cite{moore09book}, polytopes \cite{asarin00hscc}, ellipsoids \cite{kurzhanski00hscc}, zonotopes \cite{girard05hscc}) or symbolic representations (such as support functions \cite{guernic09cav}, Taylor models \cite{berz98rc,chen12rtss}) to depict sets of system states and perform set propagation to construct approximations of the reachable set.
% For a comprehensive survey on this topic, we recommend referring to \cite{althoff21survey}.
% Alternatively, simulation-based method represents system states by nearby sampled trajectories and attempt to cover the reachable set by a finite number of neighborhoods of trajectories  \cite{girard06hscc,donze07hscc,asarin07acta,duggirala13emsoft,fan16emsoft}.
% Another class of methods represents system states by constructing a finite state abstraction of the system, thereby enabling the incorporation of model checking techniques \cite{alur00ieee,baier08book,tabuada09book}.

\emph{Deductive verification}, derived from Hoare-style program verification \cite{hoare69}, offers a method to verify safety without directly computing the reachable set.
At the core of deductive verification lies the synthesis of \emph{differential invariants} \cite{liu2011emsoft,platzer08cav}, which extend the concept of inductive invariants to the continuous-time domain.
Specifically, a differential invariant is a set of states from which any trajectories starting from it can never escape.
With a priori specified template, the invariant generation problem boils down to solving the constraints encoding the invariant condition. 
When all involved constraints are polynomial, the problem is decidable but has time complexity doubly exponential in the number of variables~\cite{liu2011emsoft}, according to Tarski's theorem \cite{tarski51} and the complexity for the quantifier elimination procedure \cite{dh88}.
Consequently, considerable efforts have been dedicated to identifying differential invariants that allow for efficient synthesis.

In their seminar work \cite{prajna04hscc}, Prajna and Jadbabaie introduced the concept of barrier certificates as witnesses to safety.
Namely, a barrier certificate is a real-valued function whose zero sub-level set serves as a differential invariant, separating the set of initial states and the unsafe region.
It is important to note that, for the purpose of efficient synthesis, the barrier certificate condition strengthens the general condition of differential invariants.
Since then, various definitions of barrier certificates have been proposed, aiming to relax the original barrier certificate conditions while still allowing for efficient synthesis.
Examples of such definitions include exponential-type barrier certificates \cite{kong13cav}, Darboux-type barrier certificates \cite{zeng16emsoft}, general convex barrier certificates \cite{dai17jsc} and vector barrier certificates \cite{sogokon18fm}, and invariant barrier certificates \cite{wang22iac}.
Moreover, similar notions of barrier certificates have been developed for verification problems that involve control inputs \cite{xu15adhs,ames17tac}, disturbances \cite{wang17jssc}, stochastic dynamics \cite{prajna07tac,huang17ecs,jagtap21tac,DBLP:conf/cav/FengC00Z20}, and temporal logic specifications \cite{tichakorn16tac,murali24hscc}.
These extensions broaden the applicability of barrier certificates in various domains.
Recently, there are also works aim at generalizing the notion of $k$-inductiveness for safety verification, leading to the definitions of $t$-barrier certificates \cite{bak18adhs} and $k$-inductive barrier certificates \cite{anand21cdc,anand22hscc}. 

Sum-of-squares optimization is a well-established computational technique for synthesizing barrier certificates and has been employed in most of the works mentioned above.
Typically, the barrier certificate conditions are first encoded into constraints involving sum-of-squares polynomials.
These constraints are then translated into SDP and solved by numerical solvers.
In scenarios where the domains are bounded, one can choose to rely on either a sound characterization or a complete characterization to encode the conditions. 
The differences between these two characterizations are often overlooked, as their formulations are quite similar.
However, when dealing with systems defined over unbounded domains, the sound characterization tends to be conservative while the complete characterization can not be utilized due to the violation of the Archimedean condition in Putinar's Positivstellensatz.
In such unbounded cases, existing methods solely rely on the sound characterization, potentially leading to conservative results.

Besides sum-of-squares optimization, much effort have been devoted to incorporate other numerical methods for solving the obtained constraints, for instance, interval arithmetic \cite{gao12ijcar,gao13cade,djaballah17auto}, linear programming \cite{ben16ima}, and data-driven approaches \cite{zhao20hscc,abate21hscc,salamati22l4dc,peruffo21tacas,zhao23ecs}.

\paragraph{Contributions.}
Our main contributions are threefold:
\begin{enumerate}
    \item We explicitly clarify the connection between the soundness and the completeness of the sum-of-squares characterization of barrier certificates, which is mostly overlooked in existing works. This can be considered as a minor contribution. (See \cref{sec:problem})
    \item We utilize the \emph{homogenization} technique from \cite{huang23mp} to derive the first complete sum-of-squares characterization of polynomial barrier certificates over unbounded domains. (See \cref{sec:poly})
    \item We introduce the definition of homogenized systems and consider a specific class of non-polynomial barrier certificates with more expressive power. We also propose a complete sum-of-squares characterization for this class of non-polynomial barrier certificates. (See \cref{sec:semi})
\end{enumerate}

Finally, we implement algorithms for synthesizing barrier certificates based on the existing incomplete characterization and our two novel complete characterizations.
These algorithms are tested over a set of benchmarks with unbounded domains adapted from the literature. 
Experimental results demonstrate that the two complete characterizations are more expressive while maintaining a comparable level of efficiency. (See \cref{sec:exp})

\paragraph{Organization.}
The rest of this paper is organized as follows: 
\cref{sec:pre} introduces algebraic tools that will be used. 
\cref{sec:problem} formulates the barrier certificate synthesis problems and explains the connection between the sound and the complete characterization in the bounded case.
\cref{sec:poly} proposes the first complete characterization of polynomial barrier certificates over unbounded domains.
\cref{sec:semi} introduces the definition of homogenized systems and extends the results to a class of non-polynomial barrier certificates.
Finally, \cref{sec:exp} reports the experimental results and \cref{sec:summary} concludes the paper. 

\section{Preliminaries}
\label{sec:pre}

In this section, we fix basic notations and introduce necessary concepts concerning sum-of-squares optimization. 
For interested readers, we recommend \cite{lasserre09book,blekherman2012sdp} for a detailed treatment of this topic.

\paragraph{Basic Notations.}
% $\NatIndex{m}{n}$ represents the index set~$\{m,m+1,\dots,n\}$ for any naturals $m, n$ such that $m\leq n$. 
Let $\Nat$, $\Real$, $\Real_{\ge 0}$, and $\Real_{>0}$ denote the set of all natural numbers, the set of reals, non-negative real numbers and the set of positive real numbers, respectively.
The set of continuously differentiable functions over $\Real^n$ is denoted by $\mathcal{C}^1(\Real^n)$.
%For any $r\in \Real$, $\lceil r\rceil$ denotes the smallest integer larger than $r$.
By convention, we use boldface letters to denote vectors and vector-valued functions, 
e.g., $\seq{x}=(x_1,\dots,x_n)$ denotes a state variable and $\seq{f}=(f_1,\dots,f_n)$ denotes a vector field. 
For vectors $\seq{x}, \seq{y}\in \Real^n$, $\langle \seq{x},\seq{y}\rangle\defeq \sum_{i=1}^n x_iy_i$ represents the inner product of $\seq{x}$ and $\seq{y}$, 
and $\norm{\seq x}\defeq \sqrt{\langle \seq{x},\seq{x}\rangle}$ denotes the standard Euclidean norm.

Let $\Real[\seq x]$ denote the set of polynomials in variables $\seq x$ with real coefficients.
% $\Real^d[\seq x]$ denote the set of polynomials with degree up to $d$. 
A basic semialgebraic set $\mathcal{K}\subseteq \mathbb R^n$ is of the form $\Set{\bm{x}\in \Real^n}{p_1(\bm{x}) \triangleright 0, \dots, p_m(\bm{x}) \triangleright 0}$, 
where $p_i(\bm{x}) \in \mathbb{R}[\bm{x}]$ and $\triangleright\in \{\ge, >\}$. 
An equality $p(\seq{x})=0$ can be represented by two inequalities $p(\seq{x}) \ge 0$ and $-p(\seq{x})\ge0$.
A basic semialgebraic set is considered \emph{closed} when its defining polynomials contain only non-strict inequalities.
Semialgebraic sets are formed as unions of basic semialgebraic sets. i.e., $\bigcup_{i=1}^{n} \mathcal{K}_i$, where each $\mathcal{K}_i$ is a basic semialgebraic set.
% For any (semialgebraic) set $S\subseteq \Real^n$, $\cl{S}$ denotes the closure of $S$. 

\paragraph{Sum-of-Squares Polynomials.}
Given $S\subseteq \Real^n$, we say $p(\seq x)\in \Real[\seq x]$ is \emph{nonnegative} (resp. \emph{strictly positive}) over $S$ if $p(\seq x)\ge 0$ (resp. $p(\seq x)> 0$) for any $\seq x\in S$.
Sum-of-squares (SOS) polynomials are an important subset of globally nonnegative polynomials over $\Real^n$.
A polynomial $p(\seq x)\in \Real[\seq x]$ is said to be a \emph{sum-of-squares} polynomial if it can be expressed as $p(\bm{x})=\sum_{i=1}^m p_i(\bm{x})^2$, where $p_i(\bm{x})\in \mathbb R[\bm{x}]$ for each $i$.
% Similar to $\mathbb R[\seq x]$ and $\mathbb R^d[\seq x]$, we use $\Sigma[\seq x]$ and $\Sigma^{d}[\seq x]$ to denote the set of sum-of-squares polynomials and sum-of-squares polynomials of degree up to $d$ in variables~$\seq x$, respectively. 
We use $\Sigma[\seq x]$ to denote the set of SOS polynomials in variables~$\seq x$.

\paragraph{Putinar's Theorem.}
Given polynomials $p_1,\dots,p_m\in\Real[\seq{x}]$. 
Let $\mathcal K$ be a closed basic semialgebraic set described by
\begin{equation}\label{eq:K}
\mathcal K \defeq \Set{\seq x \in \Real^n}{p_1(\seq x)\ge~0, \dots, p_m(\seq x)\ge~0}.
\end{equation}
The set of polynomials
\begin{equation*}
    \qm(p_1, p_2, \ldots, p_m) \defeq
        \big\{ \sigma_0 + \sum_{i=1}^{m} \sigma_i p_i \mid \sigma_i \in \Sigma[\bm{x}]\text{ for }i=0,1,\dots,m \big\}
\end{equation*} 
is called the \emph{quadratic module} generated by $p_1,\dots,p_m$. 
A quadratic module $\qm$ is \emph{Archimedean}, or satisfies the \emph{Archimedean condition}, 
if $N - \norm{\bm{x}}^2 \in \qm$ for some constant $N \in \mathbb{N}$. 
Since a sum-of-squares polynomial~$\sigma(\seq x)\in \Sigma[\seq x]$ is nonnegative over $\Real^n$, the following result trivially holds.
\begin{proposition}\label{prop:qm}
    Given $\mathcal K$ as defined in \cref{eq:K}, then 
    \begin{equation*}
        f(\seq x)\in \qm(p_1,\dots,p_m) \implies f(\seq x)\ge 0 \text{ over } \mathcal K.    
    \end{equation*}
\end{proposition}

An important result in real algebraic geometry is Putinar's Positivstellensatz, which states that, under the Archimedean condition, the quadratic module $\qm(p_1,\dots,p_m)$ contains all polynomials strictly positive over $\mathcal K$.
 
\begin{theorem}[Putinar's Positivstellensatz \cite{putinar93,lasserre09book}]
\label{thm:putinar}
    Given~$\mathcal K$ as defined in \cref{eq:K} and a polynomial~$f\in \mathbb R[\bm{x}]$, 
    if $\qm(p_1,\dots,p_m)$ is Archimedean,
    then 
    \begin{equation*}
        f(\seq x)> 0 \text{ over } \mathcal K \implies f(\seq x)\in \qm(p_1,\dots,p_m).
    \end{equation*}
    %$f(\seq x)>0$ over $\mathcal K$ implies that $f\in \qm(p_1,\dots,p_m)$.
\end{theorem}

Here, the condition ``$\qm(p_1,\dots,p_m)$ is Archimedean" can be intuitively understood as $\mathcal K$ in \cref{eq:K} is bounded.
In one direction, if $\qm(p_1,\dots,p_m)$ is Archimedean, by \cref{prop:qm}, we have $N-\|\seq{x}\|^2\ge 0$ over $\mathcal K$, hence $\mathcal{K}$ is bounded.
In the other direction, when $\mathcal{K}$ is bounded within a ball $\{\seq x\in\Real^n \mid N-\norm{\seq x}^2\ge 0\}$, then we can assume a redundant constraint~$p_{m+1}=N-\norm{\seq x}^2$ and the new quadratic module $\qm(p_1,\dots,p_m,p_{m+1})$ is Archimedean.
%Note that \cref{prop:qm} is not subject to this restriction.
In general, note that \cref{prop:qm} does not necessarily imply that $\mathcal{K}$ is bounded.

\section{Problem Formulation}
\label{sec:problem}
In this section, we formally define the barrier certificate synthesis problem of interest, and discuss the relation between the sound and the complete sum-of-squares characterization of polynomial barrier certificates over bounded domains. 
The majority of the existing literature, such as \cite{prajna04hscc,kong13cav,dai17jsc,sogokon18fm}, primarily focus on the sound characterization.
As far as we are aware, the complete characterization is only mentioned in \cite{wang22iac}.
Subsequently, we clarify the connection between these two characterizations, which can be considered as a minor contribution.

\paragraph{Differential Dynamical Systems.} 
We consider a class of dynamical systems featuring
differential dynamics governed by ordinary differential equations (ODEs) of autonomous type:
\begin{equation}\label{eq:system}
    \dot{ \seq x} = \seq f (\seq x)
\end{equation}
where $\seq x\in \Real^n$ is the state vector, 
$\dot{\seq x}$ denotes its temporal derivative $dx/dt$, and $\seq f:\Real^n\to \Real^n$ is a polynomial vector field, i.e., each component $f_i$ of $\seq{f}$ is a polynomial.
%a Lipschitz continuous .
Since a polynomial vector field is locally Lipschitz continuous, ODE~\eqref{eq:system} admits an unique \emph{solution} (or \emph{trajectory}), denoted as $\traj_{\seq x_0}: \Real_{\ge0} \to \Real^n$, from any initial state $\seq x_0\in \Real^n$, such that
(1) $\traj_{\seq x_0}(0)=\seq x_0$ (2) for any $t'\in\Real_{\ge 0}$, $\frac{\mathrm d \traj_{\seq x_0}}{\mathrm d t}\big \vert_{t=t'} = \seq f(\traj_{\seq x_0}(t')).$
% \begin{enumerate}
    % \item $\traj_{\seq x_0}(0)=\seq x_0$
    % \item $\frac{\mathrm d \traj_{\seq x_0}}{\mathrm d t}\big \vert_{t=t'} = \seq f(\traj_{\seq x_0}(t')), ~\forall t'\in\Real_{\ge 0}.$
% \end{enumerate}

\paragraph{Safety Verification Problems.}
Given dynamical system \cref{eq:system} with domain $\domain \subseteq \Real^n$, initial set $\init \subset \domain$, and unsafe set $\unsafe \subset \domain$,
the \emph{safety verification problem} asks whether $\unsafe$ is reachable from any state in $\init$ within $\domain$.
Formally, let $\reach$ denote the reachable set,
\begin{equation*}
    \reach \defeq \Set{\seq x\in \domain}{
    \; \exists t\in \Real_{\ge 0}, \exists \seq x_0 \in \init.
    ~\seq x = \traj_{\seq x_0}(t)},
    % for simplicity, just delete the domain requirement here
    % \wedge \big(\forall \tau \in [0, t].~\traj_{\seq x_0}(\tau)\in \domain \big)},    
\end{equation*}
where we assume that a trajectory will never leave the domain.
The system is said to be \emph{safe} if $\unsafe \cap \reach=\emptyset$, and \emph{unsafe} otherwise. 

In this paper, we restrict our focus to the case when $\domain$, $\init$, and $\unsafe$ are closed basic semialgebraic sets described by
\begin{align*}
    \domain &\defeq \Set{\seq x \in \Real^n}{g^\domain_{i}(\seq x)\ge 0, \text{ for } i= 1,\dots, m_x},\\
    \init &\defeq \Set{\seq x\in \domain}{g^\init_{i}(\seq x)\ge 0, \text{ for } i= 1,\dots, m_i},\\
    \unsafe &\defeq \Set{\seq x\in \domain}{g^\unsafe_{i}(\seq x)\ge 0, \text{ for } i= 1,\dots, m_u}.
\end{align*}
%where $g^\init_i\in \Real[\seq x]$ and $g^\unsafe_j \in \Real[\seq x]$ for each $i$ and $j$.

% We would like to note that the safety verification problem can be readily addressed when the reachable set $\reach$ is computable.
% Nevertheless, for the majority of nonlinear systems, the direct computation, or even approximate estimation, of reachable sets typically proves intractable.
% This leads to the notion of invariant sets.

\paragraph{Invariants.}

A \emph{differential invariant} is a subset $\Phi \subseteq \domain$ such that any trajectory starting from $\Phi$ stays within $\Phi$ forever, i.e.,
\begin{equation*}
    \forall \seq x_0 \in \Phi, \forall t\in \Real_{\ge 0}.~\traj_{\seq x_0}(t)\in \Phi.
\end{equation*}
Utilizing this concept, we can verify the safety of a system without explicitly computing the reachable set, which is typically intractable for the majority of nonlinear systems.
The idea therein is to find a differential invariant~$\Phi \subseteq \domain$ such that  $\init\subseteq \Phi $ and $\unsafe\subseteq \domain\backslash \Phi$. 
According to the definition, the differential invariant $\Phi$ serves as an over-approximation of the reachable set~$\reach$, thereby substantiating safety of the system.

\paragraph{Barrier Certificates.}
Barrier certificates encapsulate the conditions requisite for a zero sub-level set of the form $\Set{\seq x\in \Real^n}{B(\seq x)\le 0}$
% \begin{equation*}
%     \Set{\seq x\in \Real^n}{B(\seq x)\le 0}, 
% \end{equation*}
to become a differential invariant, 
where $B\in \mathcal{C}^1(\Real^n)$.
% \revise{
For the ease of explanation, we focus on exponential-type barrier certificates and refer to them as barrier certificates for simplicity.
The technique presented in this paper can be readily extended to other types of barrier certificates \cite{dai17jsc,sogokon18fm,wang22iac} and hybrid systems (systems containing discrete transitions and continuous evolution) \cite{prajna04hscc}.
% , as long as \cref{thm:putinar} is used to encode the barrier certificate conditions into constraints.
% }

% \revcomment{The technique presented in this paper can be readily extended to other types of barrier certificates, Say in 1-2 sentences how this extension works.}

\begin{theorem}[Exponential-type Barrier Certificates, modified from \cite{kong13cav}]\label{thm:bc}
Given the system \eqref{eq:system} with sets~$\domain$,~$\init$, and~$\unsafe$. 
For any $\lambda\in \Real$, the system is safe if there exists an exponential-type barrier certificate, namely a real-valued function $B(\seq x)\in \mathcal C^1(\Real^n)$ satisfying the following conditions
\begin{align}
    &\forall \seq x\in \init.~B(\seq x)\le 0, \label{eq:bc-init}\\
    &\forall \seq x\in \unsafe.~B(\seq x) \ge \epsilon_e, \label{eq:bc-unsafe}\\
    &\forall \seq x\in \domain.~\Lie_{\seq f} B(\seq x) - \lambda B(\seq x)\le  0, \label{eq:bc-inv}
\end{align}
for some real constant $\epsilon_e\in \Real_{>0}$, where $\Lie_{\seq f}p(\seq x) \defeq \langle \frac{\partial}{\partial \seq x}  p(\seq x),  \seq f(\seq x)\rangle $ is the Lie derivative of $p(\seq x)$ over the vector filed $\seq{f}$.  
\end{theorem}

% Given a polynomial~$p(\seq x)\in \Real[\seq x]$, the Lie derivative of $p(\seq x)$ w.r.t. a vector filed $\seq f$ is denoted by $\Lie_{\seq f}p(\seq x) \defeq \langle \frac{\partial}{\partial \seq x}  p(\seq x),  \seq f(\seq x)\rangle $. 

The difference between our \cref{thm:bc} and its original formulation in \cite{kong13cav} lies in \cref{eq:bc-unsafe}, which was written as 
\begin{equation}\label{eq:bc-unsafe-old}
    \forall \seq x\in \unsafe.~B(\seq x) > 0. \tag{4'}
\end{equation}
When the unsafe region $\unsafe$ is bounded (compact), the two condition \cref{eq:bc-unsafe} and \cref{eq:bc-unsafe-old} coincide, as a continuous function over a compact set always attains a minimum.
However, when $\unsafe$ is unbounded, our formulation is stricter in the sense that $\init$ and $\unsafe$ can not be arbitrarily close, otherwise we would be unable to distinguish between them, as shown in the following \cref{ex:epsilon}. 
In theory, $\epsilon_e$ can be any real constant in $\Real_{>0}$, and the corresponding $B(\seq{x})$ will be equivalent up to a constant factor.
% In practice, this requirement is considered reasonable because we would like a barrier certificate to be robust, 

\begin{example}\label{ex:epsilon}
Consider a system $\seq{f}(x_1,x_2)=(x_1, 0)$ with $\domain=\Real^2$, $\init=\{(x_1,x_2)\mid x_1x_2+1\le 0, x_1\le 0\}$, and $\unsafe=\{(x_1,x_2)\mid x_1x_2-1\ge 0, x_1\ge 0\}$.
The function $B(x_1,x_2)=x_1$ is not a valid barrier certificate according to our definition, as the condition \cref{eq:bc-unsafe} is not satisfiable for any $\epsilon_e>0$ (though when $\epsilon_e=0$ \cref{eq:bc-unsafe-old} is satisfied).
In other words, the sets $\init$ and $\unsafe$ are indistinguishable in practice when $x_2$ goes to infinity, and our \cref{thm:bc} rules out such cases.
\end{example}

To ensure computational tractability, the barrier certificate $B(\seq x)$ is commonly constrained to polynomial forms. 
One of the prevailing computational methods for synthesizing $B(\seq x)\in \Real[\seq x]$ is based on the sum-of-squares optimization.
Now we present the sound and complete sum-of-squares characterizations of polynomial barrier certificate over bounded domains.

\begin{theorem}[Bounded Case]
\label{thm:bounded}
Let $\domain$, $\init$, and $\unsafe$ be bounded, i.e., the corresponding quadratic module is Archimedean.
Given $\lambda\in \Real$ and $\epsilon_e\in \Real_{>0}$,  
consider the following constraints with parameter $\epsilon$,
\begin{equation}\label{eq:bc-sos}
    \begin{aligned}
        & -B(\seq x)+\epsilon = \sigma_0^\init + \sum_{i=1}^{m_i} g_i^\init (\seq{x})\sigma_i^\init\\
        & B(\seq{x}) - \epsilon_e + \epsilon = \sigma_0^\unsafe + \sum_{i=1}^{m_u} g_i^\unsafe (\seq{x})\sigma_i^\unsafe\\
        & \lambda B(\seq x) - \Lie_{\seq f} B(\seq x) +\epsilon = \sigma_0^\domain + \sum_{i=1}^{m_x} g_i^\domain (\seq{x})\sigma_i^\domain\\
        &\sigma_0^\init,\dots, \sigma_{m_i}^\init, \sigma_0^\unsafe,\dots, \sigma_{m_u}^\unsafe, \sigma_0^\domain,\dots, \sigma_{m_x}^\domain \in \Sigma[\seq{x}].
    \end{aligned}
\end{equation}    
When $\epsilon=0$, \cref{eq:bc-sos} gives a sound characterization of polynomial barrier certificates, i.e., any solution $B(\seq{x})\in \Real[\seq{x}]$ to the above constraints is a barrier certificate.
When $\epsilon>0$, \cref{eq:bc-sos} gives a complete characterization of polynomial barrier certificates, i.e., any barrier certificate $B(\seq{x})\in \Real[\seq{x}]$ satisfies \cref{eq:bc-sos}.
\end{theorem}

\begin{proof}
    \cref{prop:qm} and \cref{thm:putinar} entail soundness and completeness, respectively.
\end{proof}

In fact, in most practical cases, \cref{eq:bc-sos} with $\epsilon=0$ can be viewed as a \emph{sound and complete} characterization.
In this situation, completeness follows from the so-called ``finite convergence property'' of \cref{thm:putinar}, which requires the underlying basic semialgebraic sets $\domain$, $\init$, and $\unsafe$ to satisfy some side conditions that are generally true \cite{nie12jsc}.
For now, we do not go deep into these details and just consider that soundness and completeness are dependent on the parameter~$\epsilon$.

Unfortunately, when the domain $\domain$ becomes unbounded, \cref{eq:bc-sos} with $\epsilon>0$ is no longer a complete characterization due to the violation of the Archimedean condition, while the $\epsilon=0$ case is still sound. 
Consequently, we can solely rely on the sound characterization to synthesize barrier certificates, which may fail to identify potential solutions as in the following example.
So the problem considered in this paper is, \textbf{can we derive a complete characterization similar to \cref{eq:bc-sos} for the unbounded cases?}

\begin{example}
    Consider an 1-dimensional system $f(x_1)=x_1$ with $\domain = \Real$, $\init=\{x_1\mid x_1^3\ge 0\}$, and $\unsafe=\{x_1\mid x_1+1\le 0 \}$, then $B(x_1)=-x_1$ is a barrier certificate but is not a solution to \cref{eq:bc-sos} with $\epsilon=0$.
    To see this, we only need to show that there exists no sum-of-squares polynomials $\sigma^\init_0(x_1), \sigma^\init_1(x_1)\in \Sigma[x_1]$ such that $x_1=\sigma^\init_0(x_1)+x_1^3\sigma^\init_1(x_1)$.
    Suppose we have such an expression,
    by setting $x_1=0$, we have $\sigma^\init_0(0)=0$.
    Assume that $\sigma^\init_0$ can be expressed as $\sigma^\init_0(x_1)= \sum_i p_i^2(x_1)$, then $\sigma^\init_0(0)=0$ implies that $p_i(0)=0$ for each $i$, so each $p_i$ factors as $p_i(x_1)=x_1 p_i'(x_1)$.
    Therefore, both $\sigma^\init_0(x_1)$ and $x_1^3\sigma^\init_1(x_1)$ contain no terms of degree less than $2$, which is impossible.
\end{example}

\section{A Complete Characterization of Polynomial Barrier Certificates}
\label{sec:poly}

In this section, we give an affirmative answer to the question raised above.
The tool we use is a newly introduced technique in the optimization community, called \emph{homogenization} \cite{huang23mp}, to transform an unbounded optimization problem into a bounded one.
In the following, we utilize the homogenization technique to derive a complete characterization for polynomial barrier certificates purely from a constraint-solving perspective.
In the next section, we will take a different view of this technique and consider a family of non-polynomial barrier certificates that arise naturally.

We first fix some notations.
Given $\seq{x}\in \Real^n$, let $x_0$ be a fresh variable.
% and denote~$\homo{\seq{x}}=(x_0,\seq{x})$.
For a polynomial~$p(\seq{x})\in\Real[\seq{x}]$ of degree~$d$, its homogenization w.r.t. variable~$x_0$ is a new polynomial $\homo{p}\in \Real[x_0,\seq{x}]$ defined by $\homo{p}(x_0, \seq x) \defeq x_0^d p(x_1/x_0,\dots,x_n/x_0)$.
For example, let $f(x_1,x_2)=x_1^2+x_2+1$, then $\homo{f}(x_0,x_1,x_2)
% =x_0^2 \cdot \big((x_1/x_0)^2+(x_2/x_0)+1\big)
=x_1^2+x_2x_0+x_0^2$.
%let $p^\infty(\seq x)$ denote the highest degree part of~$p(\seq x)$, and .
Suppose $\mathcal K\subseteq \Real^n$ is a semialgebraic set as described in~\cref{eq:K}, we introduce two related sets in $\Real^{n+1}$ as follows:
\begin{alignat*}{2}
    &\homo{\mathcal K}_{>0}&& \defeq 
    \Set{ (x_0, \seq{x})}{\homo{p}_1(x_0, \seq{x})\ge0, \dots, \homo{p}_m(x_0, \seq{x})\ge0, \norm{\seq{x}}^2+x_0^2=1, x_0> 0},\\
    &\homo{\mathcal K} &&\defeq 
    \Set{ (x_0, \seq{x})}{\homo{p}_1(x_0, \seq{x})\ge0, \dots, \homo{p}_m(x_0, \seq{x})\ge0, \norm{\seq{x}}^2+x_0^2=1, x_0\ge 0}.
\end{alignat*}
One can see that there exists an one-to-one mapping between $\homo{\mathcal K}_{>0}$ and $\mathcal{K}$:
\begin{lemma}\label{lem:homo}
Let $\mathcal K$ be as in \cref{eq:K}. Then $\seq{x} \in \mathcal{K}$ if and only if
  \begin{equation*}
    \left(\frac{1}{\sqrt{1+\|\seq{x}\|^2}}, \frac{x_1}{\sqrt{1+\|\seq{x}\|^2}}, \ldots, \frac{x_n}{\sqrt{1+\|\seq{x}\|^2}}\right)\in \homo{\mathcal K}_{>0}.
  \end{equation*}
  Moreover, $(x_0, \seq{x}) \in \homo{\mathcal K}_{>0}$ if and only if 
  $(\frac{x_1}{\sqrt{1-\|\seq{x}\|^2}}, \ldots, \frac{x_n}{\sqrt{1-\|\seq{x}\|^2}})\in \mathcal K$.
\end{lemma}

%\begin{proof}
%    Straightforward to verify.
%\end{proof}

% the projection map
% \begin{equation*}
%     \varphi: \Set{(x_0, \seq{x})}
%     {x_0>0, \norm{\seq{x}}^2+x_0^2=1} \to \Real^n, (x_0,\seq x)\mapsto \frac{\seq x}{x_0}
% \end{equation*} 
% defines an one-to-one mapping between points in $\homo{\mathcal K}$ with $x_0\neq 0$ and $\mathcal K$.
% By employing the inverse mapping~$\varphi^{-1}$, 

Utilizing the above lemma, we can transform a potentially \emph{unbounded} set into a \emph{bounded} set located on the unit sphere within $\Real^{n+1}$. 
Moreover, note that points with $x_0=0$ in $\Real^{n+1}$ correspond to points at infinity in $\Real^n$. 
This encourages us to take the points at infinity into consideration.
The related concept is captured by the following definition.

% \sfcomment{suggest to move this definition to thm4, and add remarks about this condition}
% I think it is better to state it as an independent definition
\begin{definition}[closed at infinity \cite{nie12jsc}] 
A basic semialgebraic set $\mathcal K$ is closed at infinity if
$cl(\homo{\mathcal K}_{>0})=\homo{\mathcal K}$,
%\begin{equation*}
%    cl(\homo{\mathcal K}_{>0})=\homo{\mathcal K},
%\end{equation*}
where $cl(\homo{\mathcal K}_{>0})$ denotes the closure of $\homo{\mathcal K}_{>0}$.
\end{definition}

We would like to emphasize that being closed at infinity is a generic property for semialgebraic sets \cite{guo14jgo}, and its manifestation may be contingent upon the selection of descriptive polynomials.
%as shown in \cref{ex:closed}. 
For example, 
let $S_1 = \Set{(x_1,x_2)}{x_1-x_2^2\ge 0}$, then $S_1$ is not closed at infinity because 
    \begin{equation*}
        (0,-1,0) \not\in cl(\homo{S_1}_{>0})\quad \text{ and }\quad 
        (0,-1,0) \in \homo{S_1}. 
    \end{equation*}
However, by adding a redundant polynomial inequality $x_1\ge 0$ in $S_1$, we can check $S_2= \Set{(x_1,x_2)\in\Real^2}{x_1-x_2^2\ge 0, x_1\ge 0} (=S_1)$ is closed at infinity.
In this paper, we assume that  $\init$, $\unsafe$, and $\domain$ are all closed at infinity, which is purely a technical assumption.
To check whether a semialgebraic set is closed at $\infty$, one can use \cite[Thm.~2.11]{guo14jgo}.
% \sfcomment{Does there exists any results that automatically check ``close at infinity"?}

% \begin{example}\cite{huang23mp}\label{ex:closed}
    % Let $S_1 = \Set{(x_1,x_2)\in\Real^2}{x_1-x_2^2\ge 0}$, then $S_1$ is not closed at infinity because 
    % \begin{equation*}
    %     (0,-1,0) \not\in cl(\homo{S_1}_{>0})\quad \text{ and }\quad 
    %     (0,-1,0) \in \homo{S_1}. 
    % \end{equation*}
    % However, by adding a redundant polynomial inequality $x_1\ge 0$ in $S_1$, we can check $S_2= \Set{(x_1,x_2)\in\Real^2}{x_1-x_2^2\ge 0, x_1\ge 0} (=S_1)$ is closed at infinity.    
% \end{example}

The following theorem lies at the core of the homogenization technique.
\begin{theorem}[{\cite[Lem~3.2]{huang23mp}}]\label{thm:homo}
    When a basic semialgebraic set $\mathcal K$ is closed at infinity, for any polynomial $f\in \Real[\seq{x}]$ 
    \begin{equation*}
        f(\seq x)\ge 0 \text{ over } \mathcal K \iff \homo{f}(x_0, \seq{x}) \ge 0 \text{ over } \homo{\mathcal K}.
    \end{equation*}
\end{theorem}

% \begin{theorem}[{\cite[Lem~3.2]{huang23mp}}]\label{thm:homo}
%     When $\mathcal K$ is closed at $\infty$, 
%     \begin{equation*}
%         f(\seq x)\ge 0 \text{ over } \mathcal K \iff \homo{f}((x_0, \seq{x}) ) \ge 0 \text{ over } \homo{\mathcal 
%     K}
%     \end{equation*}
%     % $f(\seq x)\ge 0$ over $\mathcal K$ if and only if $\homo{f}((x_0, \seq{x}) ) \ge 0$ over $\homo{\mathcal 
%     % K}$.  
% \end{theorem}

Now we present the homogenized version of \cref{thm:bounded}, which solves the problem raised at the end of the last section.

\begin{theorem}\label{thm:unbounded-poly}
Assume that $\init$, $\unsafe$, and $\domain$ are all closed at infinity.
Given $\lambda\in \Real$ and $\epsilon_e\in \Real_{>0}$,  
consider the following constraints with parameter $\epsilon$,
\begin{equation}
    \begin{aligned}\label{eq:unbound-poly}
        & - \homo{B}(x_0, \seq{x}) + \epsilon = \sigma_0^\init + \sum_{i=1}^{m_i+2} \sigma^\init_i \homo{g}^\init_i\\
        & \homo{B}(x_0, \seq{x}) - \epsilon_e x_0^{d} + \epsilon = \sigma_0^\unsafe + \sum_{i=1}^{m_u+2} \sigma^\unsafe_i \homo{g}^\unsafe_i \\
        &  \homo{H}(x_0, \seq{x}) + \epsilon =  \sigma_0^\domain + \sum_{i=1}^{m_x+2} \sigma^\domain_i \homo{g}^\domain_i\\
        & \sigma^\init_0, \dots, \sigma^\init_{m_i+1}, 
        \sigma^\unsafe_0, \dots, \sigma^\unsafe_{m_u+1},
        \sigma^\domain_0, \dots, \sigma^\domain_{m_x+1}
        \in \Sigma[x_0, \seq{x}],\\ &
        \sigma^\init_{m_i+2}, \sigma^\unsafe_{m_u+2},
        \sigma^\domain_{m_x+2} \in \Real[x_0, \seq{x}],
    \end{aligned}
\end{equation}
where $H(\seq x)\defeq \lambda B(\seq x) - \Lie_{\seq f} B(\seq x)$, $d$ is the degree of $\deg{B}(\seq{x})$, $\homo{g}_{m_i+1}^\init = \homo{g}_{m_u+1}^\unsafe = \homo{g}_{m_x+1}^\domain = x_0$, and $\homo{g}_{m_i+2}^\init = \homo{g}_{m_u+2}^\unsafe = \homo{g}_{m_x+2}^\domain = x_0^2 + \norm{\seq{x}}^2-1$.   
When $\epsilon=0$, \cref{eq:unbound-poly} gives a sound characterization of polynomial barrier certificates, i.e., any solution $B(\seq{x})\in \Real[\seq{x}]$ of degree~$d$ to the above constraints is a barrier certificate.
When $\epsilon>0$, \cref{eq:unbound-poly} gives a complete characterization of polynomial barrier certificates, i.e., any barrier certificate $B(\seq{x})\in \Real[\seq{x}]$ of degree $d$ satisfies the above constraints.
\end{theorem}

\begin{proof}
    We prove the first constraint corresponding to the initial set $\init$, the other two constraints are similar. 
    By employing homogenization and \cref{thm:homo}, the original condition \cref{eq:bc-init} can be transformed into $- \homo{B}(x_0, \seq{x})\ge 0$ over $\homo{\init}$. 
    Since the descriptive polynomials in $\homo{\init}$ contain $\|\seq{x}\|^2+x^2_0 = 1$, $\homo{\init}$ is a closed basic semialgebraic set and its corresponding quadratic module is Archimedean.
    Thus, we can apply \cref{prop:qm} and \cref{thm:putinar} to obtain the soundness and completeness results, respectively.
    Note for the other two constraints, we need to homogenize the polynomial $B(\seq{x})-\epsilon_e$ and $B(\seq x) - \Lie_{\seq f} B(\seq x)$ as a whole.
\end{proof}

\section{Homogenized Systems and Semialgebraic Barrier Certificates}
\label{sec:semi}

In this section, we take a different perspective of the technique in the last section.
The motivation comes from the observation that the homogenization procedure can be viewed as mapping the original system in $\Real^n$ into a new system in $\Real^{n+1}$.
Consequently, the constraints in \cref{eq:unbound-poly} can be conceived as barrier certificate conditions for the new system.
Employing this idea, we introduce the definition of homogenized systems as follows. 
To avoid confusion, we will use $(y_0, \seq{y})\in \Real^{n+1}$ to denote the state variables of the homogenized systems.

\begin{definition}[Homogenized System]
Given a system~\cref{eq:system}, the homogenized system is an associated system in $\Real^{n+1}$.
For each state $\seq{x}\in \Real^n$ of the original system, the corresponding state $(y_0,\seq{y})\in \Real^{n+1}$ of the homogenized system is given by
\begin{equation} \label{eq:system-homo}
    (y_0, y_1, \dots, y_n) = (\frac{1}{\sqrt{1+\|\seq{x}\|^2}}, \frac{x_1}{\sqrt{1+\|\seq{x}\|^2}}, \ldots, \frac{x_n}{\sqrt{1+\|\seq{x}\|^2}}).
\end{equation}
\end{definition}

The dynamics of the homogenized systems can be obtained by taking derivative in the right-hand-side of \cref{eq:system-homo}. 
Hence, the safety verification problem of the original system~\cref{eq:system} with sets $\domain$, $\init$, and $\unsafe$ can be translated into an equivalent problem for the homogenized system~\cref{eq:system-homo} with sets $\homo{\domain}$, $\homo{\init}$, and $\homo{\unsafe}$.
Furthermore, we show that a barrier certificate of the original system can be computed from a barrier certificate of the homogenized system. 

\begin{theorem}\label{thm:equiv}
    $B(y_0, \seq{y})\in \mathcal C^1(\Real^{n+1})$ is a barrier certificate of the homogenized system if and only if $B(\frac{1}{\sqrt{\|\seq{x}\|^2+1}}, \frac{\seq{x}}{\sqrt{\|\seq{x}\|^2+1}})$ is a barrier certificate of the original system.
\end{theorem}

\begin{proof}
    Let $B(y_0, \seq{y})$ be a barrier certificate of the homogenized system. Denote $g(\seq{x})\defeq B(\frac{1}{\sqrt{\|\seq{x}\|^2+1}}, \frac{\seq{x}}{\sqrt{\|\seq{x}\|^2+1}})$, we show that $g(\seq{x})$ satisfies the conditions in \cref{thm:bc}.
    For $\seq{x}\in\init$, since $(y_0, \seq{y})\in \homo{\init}^b$ by \cref{lem:homo} and \cref{eq:system-homo},  we have $g(\seq{x})=B(y_0,\seq{y})\le 0$.
    Similarly, for $\seq{x}\in \unsafe$, we have $g(\seq{x})= B(y_0,\seq{y})\ge \epsilon_e$.
    Finally, since 
    \begin{equation*} 
        \begin{aligned}
            \Lie_{\seq{f}} g(\seq {x}) 
            &= \sum_{i=1}^n \frac{\partial g(\seq x)}{\partial x_i} f_i(\seq{x}) = \sum_{i=1}^n \left(\sum_{j=0}^n \frac{\partial B(y_0, \seq{y})}{\partial y_j} \frac{\partial y_j}{\partial x_i}\right) f_i(\seq{x})\\
            & = \sum_{j=0}^n \frac{\partial B(y_0, \seq{y})}{\partial y_j} \left(\sum_{i=1}^n \frac{\partial y_j}{\partial x_i} f_i(\seq{x})\right) = \Lie_{\seq{f}'} B(y_0,\seq{y}),
        \end{aligned}
    \end{equation*}
    where $\seq{f}'$ is the dynamic of the homogenized system. For any $\seq{x}\in \domain$ we have $\Lie_{\seq{f}} g(\seq{x})-\lambda g(\seq{x}) = \Lie_{\seq{f}'} B(y_0, \seq{y})-\lambda B(y_0,\seq{y}) \le 0$.
    The other direction is similar.
\end{proof}

According to Stone–Weierstrass theorem \cite{stone48}, a continuous function in a compact space in $\Real^{n+1}$ can be approximated by polynomials. 
% \revcomment{It would be helpful to show how the accuracy of these barrier certificate approximations changes w.r.t. the dimension of the system and degree of the approximation. }
This means that, if there exists $B(y_0, \seq{y})\in \mathcal C^1(\Real^{n+1})$ as a barrier certificate, one should be able to find a polynomial barrier certificate (of sufficient large degree) close to it. 
% we can usually find a polynomial barrier certificate close to it.
In fact, this is one of the reasons why we are primarily concerned with polynomial barrier certificates in the bounded case.
By \cref{thm:equiv}, if $B(y_0, \seq{y})$ is a polynomial of degree $d$, then we have
\begin{equation*}
    (\sqrt{\|\seq{x}\|^2+1})^d B(\frac{1}{\sqrt{\|\seq{x}\|^2+1}}, \frac{\seq{x}}{\sqrt{\|\seq{x}\|^2+1}}) =  B_1(\seq{x}) + \sqrt{\|\seq{x}\|^2+1} \cdot B_2(\seq{x})
\end{equation*}
for some polynomials $B_1(\seq{x}),B_2(\seq{x})\in \Real[\seq{x}]$.
From this expression, we can see that \cref{thm:unbounded-poly} is a special case when $B(y_0, \seq{y})$ itself is a homogeneous polynomial (i.e., all monomials are of the same degree) and $B_2(\seq{x})=0$.
% \revcomment{does not seem that the semi-algebraic barrier certificate can easily express any semi-algebraic set. If that is a subclass, then the paper should find a different name.}

\begin{definition}
    We say a barrier certificate $B(\seq{x})$ is \emph{semialgebraic}
    \footnote{A function $f(\seq{x})$ is called semialgebraic if its graph $\{(\seq{x},f(\seq{x}))\mid \seq{x}\in \Real^n\}$ is a semialgebraic set. Here, semialgebraic barrier certificates can only express a certain subclass of semialgebraic functions.}
    if it can be expressed as $B(\seq{x})=B_1(\seq{x}) + \sqrt{\|\seq{x}\|^2+1} \cdot B_2(\seq{x})$ for some $B_1(\seq{x}),B_2(\seq{x})\in \Real[\seq{x}]$.
\end{definition}

The synthesis of semialgebraic barrier certificates is not straightforward, due to the existence of non-polynomial component $\sqrt{\|\seq{x}\|^2+1}$. 
To address this problem, we employ the technique in \cite{lasserre2012positivity} to encode these non-polynomial expressions into polynomials with extra variables.
To be concrete, we introduce two variables $u$ and $v$, which stand for $\sqrt{\|\seq{x}\|^2+1}$ and $\frac{1}{\sqrt{\|\seq{x}\|^2+1}}$, respectively.
Then, by \cref{thm:bc}, the conditions for a semialgebraic barrier certificate can be written as
\begin{equation}\label{eq:bc-semi}
\begin{aligned}
    B_1(\seq x)+u B_2(\seq x)&\le 0, &&\text{ for } \seq x\in \init, u^2=\|\seq{x}\|^2+1, u\ge 0, \\
    B_1(\seq x)+u B_2(\seq x)&\ge \epsilon_e, &&\text{ for } \seq x\in \unsafe, u^2=\|\seq{x}\|^2+1, u\ge 0,\\
    G(\seq{x},u,v)&\ge 0, &&\text{ for } \seq x\in \domain, u^2=\|\seq{x}\|^2+1, u\ge 0, uv=1 ,
\end{aligned}    
\end{equation}
where $G(\seq{x},u,v)\in \Real[\seq{x},u,v]$ is defined by 
\begin{equation}\label{eq:G}
\begin{aligned}
&\lambda \left(B_1(\seq{x})+\sqrt{\|\seq{x}\|^2+1} \cdot B_2(\seq{x})\right) - \Lie_{\seq{f}}\left(B_1(\seq{x})+\sqrt{\|\seq{x}\|^2+1} \cdot B_2(\seq{x})\right) \\
    % = & \lambda \left(B_1(\seq{x})+ u \cdot B_2(\seq{x})\right) - \Lie_{\seq{f}} B_1(\seq{x}) - \sum_{i=1}^n f_i(\seq{x}) \cdot \left( \frac{x_i}{\sqrt{\|\seq{x}\|^2+1}}B_2(\seq{x})+ \sqrt{\|\seq{x}\|^2+1}\frac{\partial B_2(\seq{x})}{\partial x_i} \right)\\
    = ~& \lambda \left(B_1(\seq{x})+ u \cdot B_2(\seq{x})\right) - \Lie_{\seq{f}} B_1(\seq{x}) - u\cdot  \Lie_{\seq{f}} B_2(\seq{x}) -  v B_2(\seq{x}) \sum_{i=1}^n x_i f_i(\seq{x}) \\
    \defeq &G(\seq{x},u,v).
\end{aligned}
\end{equation}

Similar to \cref{thm:bounded} and \cref{thm:unbounded-poly}, we have the following characterization for semialgebraic barrier certificates.
Without loss of generality, we assume that $B_1(\seq{x})$ and $B_2(\seq{x})$ are both of degree $d$.

\begin{theorem}\label{thm:unbounded-semi}
Assume that $\init$, $\unsafe$, and $\domain$ are all closed at infinity.
Given $\lambda\in \Real$ and $\epsilon_e\in \Real_{>0}$,  
consider the following constraints with parameter $\epsilon$,
\begin{equation}
    \begin{aligned}\label{eq:unbound-semi}
        & B(\seq{x}, u) = B_1(\seq{x}) + u \cdot B_2(\seq{x})\\
        & - \homo{B}(x_0, \seq{x}, u) + \epsilon = \sigma_0^\init + \sum_{i=1}^{m_i+4} \sigma^\init_i \homo{g}^\init_i \\
        & \homo{B}(x_0, \seq{x}, u) - \epsilon_e x_0^{d+1} + \epsilon = \sigma_0^\unsafe + \sum_{i=1}^{m_u+4} \sigma^\unsafe_i \homo{g}^\unsafe_i \\
        &  \homo{G}(x_0, \seq{x}, u, v) + \epsilon =  \sigma_0^\domain + \sum_{i=1}^{m_x+5} \sigma^\domain_i \homo{g}^\domain_i\\
        & \sigma^\init_0, \dots, \sigma^\init_{m_i+2}, 
        \sigma^\unsafe_0, \dots, \sigma^\unsafe_{m_u+2}
        \in \Sigma[x_0, \seq{x}, u], \\
        &\sigma^\domain_0, \dots, \sigma^\domain_{m_x+2}
        \in \Sigma[x_0, \seq{x}, u, v],\\ 
        &\sigma^\init_{m_i+3}, \sigma^\init_{m_i+4}, \sigma^\unsafe_{m_u+3}, \sigma^\unsafe_{m_u+4}\in \Real[x_0, \seq{x}, u]\\
        &\sigma^\domain_{m_x+3}, \sigma^\domain_{m_x+4}, \sigma^\domain_{m_x+5} \in \Real[x_0, \seq{x},u,v],
    \end{aligned}
\end{equation}
where $G(\seq x,u,v)$ is as defined in \cref{eq:G}, $\homo{g}_{m_i+1}^\init = \homo{g}_{m_u+1}^\unsafe = \homo{g}_{m_x+1}^\domain = x_0$, 
$\homo{g}_{m_i+2}^\init = \homo{g}_{m_u+2}^\unsafe = \homo{g}_{m_x+2}^\domain = u$,
$\homo{g}_{m_i+3}^\init = \homo{g}_{m_u+3}^\unsafe = \homo{g}_{m_x+3}^\domain = u^2 - x_0^2 -\norm{\seq{x}}^2$,
$\homo{g}_{m_i+4}^\init = \homo{g}_{m_u+4}^\unsafe= x_0^2 + \norm{\seq{x}}^2 + u^2 - 1$,
$\homo{g}_{m_x+4}^\domain = uv-x_0^2$, and
$\homo{g}_{m_x+5}^\domain = x_0^2 + \norm{\seq{x}}^2 + u^2 + v^2 - 1$.   
When $\epsilon=0$, \cref{eq:unbound-semi} gives a sound characterization for semialgebraic barrier certificates, i.e., any pair of solutions $B_1(\seq{x}), B_2(\seq{x})\in \Real[\seq{x}]$ to the above constraints makes $B(\seq{x})$ a barrier certificate.
When $\epsilon>0$, \cref{eq:unbound-semi} gives a complete characterization for semialgebraic barrier certificates, i.e., any semialgebraic barrier certificate with $B_1(\seq{x}),B_2(\seq{x})\in \Real[\seq{x}]$ of degree $d$ satisfies the above constraints.
\end{theorem}

\begin{proof}
    By applying \cref{thm:homo} to \cref{eq:bc-semi}. Similar to the proof of \cref{thm:unbounded-poly}.
\end{proof}

\section{Experiments}
\label{sec:exp}

\paragraph{Implementation.}
We implemented the barrier certificate synthesis procedures in \tool{Julia} programming language, interfaced with \textsc{TSSOS} \cite{tssos} for formulating SOS relaxations and \tool{Mosek} solver \cite{mosek} for solving the underlying SDP.
All experiments were performed on a Mac lap-top with Apple M2 chip and 8GB memory.
The code and benchmarks are publicly available online\footnote{\url{https://github.com/EcstasyH/BCunbounded}}.
In the following, we use the corresponding theorems to refer to different approaches/characterizations.

\paragraph{Experiment Settings.}
The goal of our experiments was to compare the differences between employing characterizations Thm.~\ref{thm:bounded}, Thm.~\ref{thm:unbounded-poly}, and Thm.~\ref{thm:unbounded-semi} to synthesize exponential-type barrier certificates over unbounded domains.
To this end, we collected a set of dynamical systems of dimension 2 and 3 from the literature.
For each benchmark system, we designed two problem instances. 
In the first instance, we only let the domain~$\domain=\Real^n$ be unbounded, 
while in the second instance, we further let the initial set~$\init$ and/or the unsafe region~$\unsafe$ be unbounded (not necessarily containing the original bounded counterparts). 
% For each problem instance, we searched for barrier certificates from degree~1 and reported the minimum degree such that either \pref{opt:gbc-suf-sos} or \pref{opt:gbc-nec-sos} is solvable.

In practical computation, we set $\lambda=-1$, $\epsilon_e = 10^{-5}$ in the definition of barrier certificates and $\epsilon=0$ in the sum-of-squares characterizations.
As discussed after Thm.~\ref{thm:bounded}, the $\epsilon=0$ case can be viewed as both sound and complete in most practical situations.
We manually verified that the sets $\init$, $\unsafe$, and $\domain$ are closed at infinity.

For Thm.~\ref{thm:bounded} and Thm.~\ref{thm:unbounded-poly}, we searched for polynomial barrier certificates $B(\seq{x})$ up to degree $6$. 
For Thm.~\ref{thm:unbounded-semi}, due to the $\sqrt{\|\seq{x}\|^2+1}$ term, we searched for semialgebraic barrier certificates with $B_1(\seq{x})$ and $B_2(\seq{x})$ up to degree 4.
When the target degree $d$ is fixed, by restricting the highest degree of involved polynomials to be the smallest even number larger than $d$, the sum-of-squares characterizations can be solved as SDPs \cite{blekherman2012sdp}. 
For each solution returned by SDP solver, we utilized \tool{Mathematica} to symbolically verify that the numerical solution $B(\seq x)$ satisfies the barrier certificate conditions.
The timeout for verifying each barrier certificate candidate was set to be 10 minutes. 
We report the total time for solving SDP constraints and verifying the results.
% The verification time is not presented but can be found through our link.\revcomment{the paper report only the synthesis time with SOS but omits the total time required to validate the barrier certificate}

\begin{table*}[t]
\captionsetup{font={small}}
\caption{Experimental results for synthesizing exponential type barrier certificates.}
\label{tab:exp}
\begin{center}
% name dim degree time verified
    \begin{tabular}{c c c ccr c ccr c ccr} 
    \toprule ~ & ~ & ~ &\multicolumn{3}{c}{Thm.~\ref{thm:bounded} \cite{kong13cav}} & ~ & \multicolumn{3}{c}{Our Thm.~\ref{thm:unbounded-poly}}& ~ & \multicolumn{3}{c}{Our Thm.~\ref{thm:unbounded-semi}}\\ 
    \cmidrule{4-6} \cmidrule{8-10} \cmidrule{12-14} 
    system & dim & unbounded & $\deg$ & succ & time(s) & ~ & $\deg$ & succ & time(s) & ~ & $\deg$ & succ & time(s)  \\
    \midrule
    \textsf{vector}\cite{sogokon18fm} & 2 &$\domain$ & 4 & \mycheckmark & 0.58  & ~  & 3 & \mycheckmark & 0.03 & ~  & \textbf{2} & \mycheckmark & 0.83 \\
    ~ & ~ & $\init,\unsafe,\domain$ & $>6$ & \mycrossmark & 0.39  & ~  &  4 & \mycheckmark & 0.16 & ~  & \textbf{2} & \mycheckmark & 0.72 \\ 
    \textsf{barrier}\cite{prajna04hscc} & 2 & $\domain$ & $>6$ & \mycrossmark & 2.40 & ~ & $>6$ & \mycrossmark & 2.73 & ~ & $>$4 & \mycrossmark & 22.98  \\
     ~ & ~ & $\init,\unsafe,\domain$ & $>6$ & \mycrossmark & 0.78 & ~ & 3 & \mycheckmark & 0.04 & ~ & \textbf{2} & \mycheckmark & 3.12  \\
    \textsf{lie-der}\cite{liu2011emsoft} & 2 & $\domain$ & 3 & \mycheckmark & 0.19 & ~ & 3 & \mycheckmark & 0.14 & ~ & \textbf{1} & \mycheckmark & 0.75  \\
    ~ & ~ & $\init,\unsafe,\domain$ & \textbf{3} & \mycheckmark & 0.14 & ~ & \textbf{3} & \mycheckmark & 0.19 & ~ & \textbf{3} & \mycheckmark & 5.04 \\
    \textsf{arch1}\cite{sogokon2016arch} & 2 & $\domain$ & \textbf{4} & \mycheckmark & 0.40 & ~ & \textbf{4} & \mycheckmark & 0.54 & ~ & $>4$ & \mycrossmark & 117.61  \\
    ~ & ~ & $\init,\unsafe,\domain$ & \textbf{1} & \mycheckmark & 0.09 & ~ & \textbf{1} & \mycheckmark & 0.04 & ~ & 2 & \mycheckmark & 20.92  \\
    \textsf{arch2}\cite{sogokon2016arch} & 2 & $\domain$ & \textbf{3} & \mycheckmark & 0.20 & ~ & \textbf{3} & \mycheckmark & 0.21 & ~ & \textbf{3} & \mycheckmark & 5.53  \\
    ~ & ~ & $\init,\unsafe,\domain$ & 3 & \mycheckmark & 0.18 & ~ & 3 & \mycheckmark & 0.18 & ~ & \textbf{2} & \mycheckmark & 1.59  \\
    \textsf{arch3}\cite{sogokon2016arch} & 2 & $\domain$ & \textbf{2} & \mycheckmark & 0.12 & ~ & \textbf{2} & \mycheckmark & 0.13 & ~ & \textbf{2} & \mycheckmark & 3.45  \\
    ~ & ~ & $\init,\unsafe,\domain$ & $>6$ & \mycrossmark & 0.84 & ~ & $>6$ & \mycrossmark & 1.30 & ~ & \textbf{1} & \mycheckmark & 0.34  \\
    \textsf{arch4}\cite{sogokon2016arch} & 2 & $\domain$ & $>6$ & \mycrossmark & 0.11 & ~ & 5 & \mycheckmark & 0.74 & ~ & \textbf{3} & \mycheckmark & 5.69  \\
    ~ & ~ & $\unsafe,\domain$ & $>6$ & \mycrossmark & 1.02 & ~ & \textbf{6} & \mycheckmark & 1.21 & ~ & $>4$ & \mycrossmark & 7.74 \\
    \textsf{nagumo}\cite{sassi2014cdc} & 2 & $\domain$ & \textbf{2} & \mycheckmark & 0.13 & ~ & \textbf{2} & \mycheckmark & 0.14 & ~ & \textbf{2} & \mycheckmark & 3.50  \\
    ~ & ~ & $\unsafe,\domain$ & $>6$ & \mycrossmark & 1.02 & ~ & \textbf{3} & \mycheckmark & 0.17 & ~ & $>4$ & \mycrossmark & 23.73  \\
    \textsf{lorenz}\cite{djaballah17auto} & 3 & $\domain$ & 6 & \myquestionmark & TO & ~ & \textbf{4} & \mycheckmark & 72.13 & ~ & 2 & \myquestionmark & TO  \\
    ~ & ~ & $\unsafe,\domain$ & \textbf{5} & \mycheckmark & 248.88 & ~ & 6 & \myquestionmark & TO & ~ & 2 & \myquestionmark & TO  \\
    \textsf{lotka}\cite{goubault2014acc} & 3 & $\domain$ & $>6$ & \mycrossmark & 88.97 & ~ & $>6$ & \mycrossmark & 27.8 & ~ & 3 & \myquestionmark & TO  \\
    ~ & ~ & $\unsafe,\domain$ & $>6$ & \mycrossmark & 0.27 & ~ & $>6$ & \mycrossmark & 438.54 & ~ & 3 & \myquestionmark & TO  \\
    \bottomrule
    \end{tabular}
    \end{center}
    \scriptsize{
    \textbf{dim}: system dimension; \textbf{unbounded}: the unbounded region(s); \textbf{deg}: degree of polynomial barrier certificates or polynomial components in semialgebraic barrier certificates; \textbf{succ}: whether our algorithm succeeds in finding a valid barrier certificate.
    \mycheckmark means valid solution, \mycrossmark means no solution or invalid solution (within the search range),  and \myquestionmark means unverified; 
    % \revise{
    TO: verification takes more than 10 minutes in \tool{Mathematica};
    \textbf{time}: total time for SDP solving and verification.
    % } 
    }
\end{table*}

\paragraph{Empirical Observations.}
\cref{tab:exp} reports the experimental results, and \cref{fig:plot} portraits selected examples. We mainly compare the results from two aspects. 

\textbf{Expressiveness}:
For problems with unbounded domains, both our complete characterizations Thm.~\ref{thm:unbounded-poly} and Thm.~\ref{thm:unbounded-semi} are more expressive than the incomplete characterization Thm.~\ref{thm:bounded}, as they succeeds in synthesizing barrier certificates in more problem instances.
The two complete characterizations offer distinct advantages:
Thm.~\ref{thm:unbounded-poly} exhibits broader applicability, demonstrably successful for problem instances like \textsf{arch4-2} and \textsf{nagumo-2}. 
In contrast, Thm.~\ref{thm:unbounded-semi} excels at synthesizing lower-degree barrier certificates, as exemplified by \textsf{vector-1,2} and \textsf{barrier-1,2} problem instances.
The experimental results also demonstrate that, while Thm.~\ref{thm:unbounded-semi} theoretically subsumes Thm.~\ref{thm:unbounded-poly}, its characterization presents significantly greater complexity and hinders its ability to identify solutions, due to inherent numerical issues in SDP solvers.

\textbf{Efficiency}:
For most benchmarks, the time overhead for employing \cref{thm:unbounded-poly} is comparable to \cref{thm:bounded}, while \cref{thm:unbounded-semi} is evidently slower than the other two.
% This should be attributed to the fact that, 
%while \cref{thm:unbounded-poly} and \cref{thm:unbounded-semi} can identify valid barrier certificates of lower degrees, 
This should be attributed to the introduction of fresh variables in SOS characterizations in the complete characterizations (one for \cref{thm:unbounded-poly} and two for \cref{thm:unbounded-semi}).
%will introduce fresh variables in SOS characterizations (one for \cref{thm:unbounded-poly} and two for \cref{thm:unbounded-semi}) and 
Hence, the computation cost of both SDP solving and posterior verification increases,   
mildly for Thm.~\ref{thm:unbounded-poly} (e.g., \textsf{lie-der-2} and \textsf{arch1-1}) but severely for Thm.~\ref{thm:unbounded-semi} (e.g., \textsf{barrier-1} and \textsf{arch1-1}). 
We also want to emphasize that, for 3-dimensional systems with higher-degree templates, posterior verification time increases significantly, meaning that we can not verify the validity of the barrier certificate candidates within a reasonable amount of time.

%For smaller systems and lower-degree templates, solving the constraints remains efficient with comparable times. 

% % It is evident that employing the complete characterizations will increase the computational cost of SDP solving for almost all benchmarks. However, the cost increase varies: mild for Thm.~\ref{thm:unbounded-poly} but severe for Thm.~\ref{thm:unbounded-semi} (e.g., \textsf{arch1}). 
% This aligns with our initial expectations, as the homogenization formulation introduces a fresh variable $x_0$ and Thm.~\ref{thm:unbounded-semi} further adds two variables $u$, $v$ to encode non-polynomial terms.
% For smaller systems and lower-degree templates, solving the constraints remains efficient with comparable times. 
% Nevertheless, currently, the efficiency loss from these factors is not a major bottleneck: 
% For larger systems and higher-degree templates, posterior verification time increases significantly, meaning that we can not decide whether the barrier certificate is valid.

\paragraph{Summary.}
For practical applications, we recommend employing Thm.~\ref{thm:unbounded-poly} to synthesize polynomial barrier certificates for unbounded problems. 
This approach achieves a high level of expressiveness while maintaining efficiency comparable to Thm.~\ref{thm:bounded}.
Moreover, we believe that the performance of Thm.~\ref{thm:unbounded-semi} can be improved by exploiting algebraic structures of the constraints. 
For example, the variables $u$, $v$ only occur linearly or quadratically in constraints, which can be utilized in restrict the templates of unknown sum-of-squares polynomials.

\begin{remark}
In our experiments, we did \emph{not} consider different parameter settings (such as the selection of $\lambda$ discussed in \cite{kong13cav}) and constraint formulations (such as techniques for taming numerical errors discussed in \cite{roux18fmsd}), which may impact the synthesized barrier certificates but are not the focus of the current paper.     
\end{remark}

\begin{figure*}[t]
    \captionsetup{font={small}}
    \centering
    \begin{subfigure}[b]{0.24\textwidth}
        \centering
        \includegraphics[width=\textwidth]{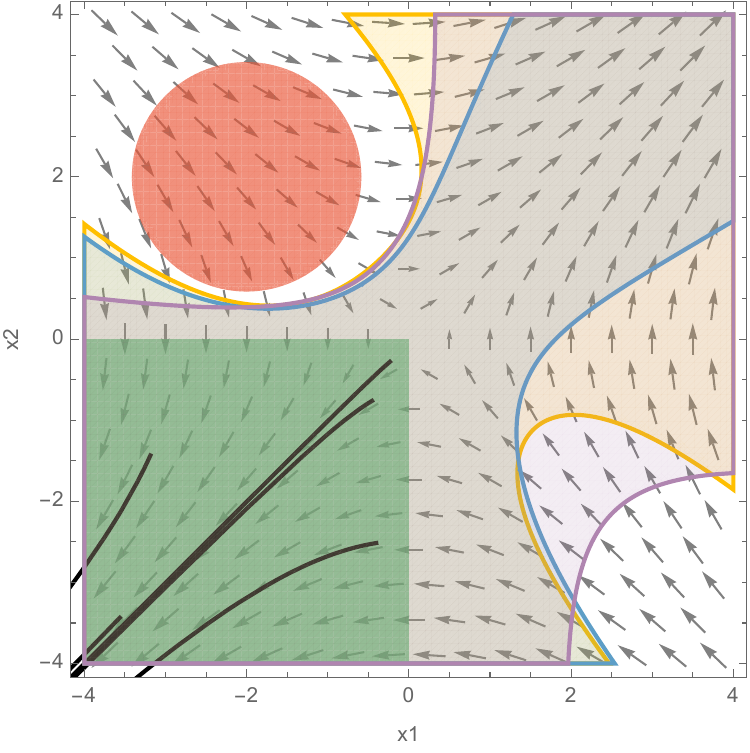}
        \caption{$\textsf{vector-1}$}
    \end{subfigure}
    \hfill
    \begin{subfigure}[b]{0.24\textwidth}
        \centering
        \includegraphics[width=\textwidth]{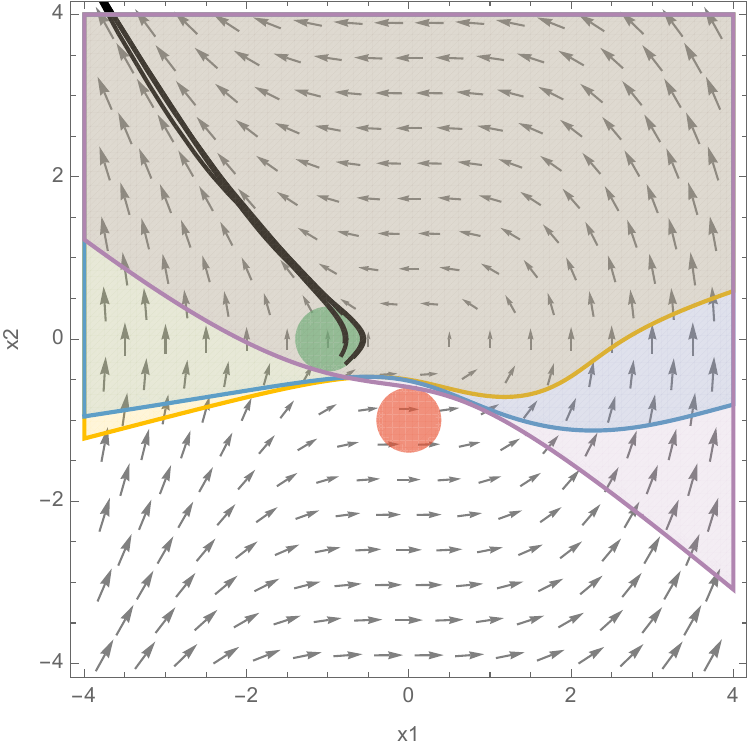}
        \caption{$\textsf{lie-der-1}$}
    \end{subfigure}
    \hfill
    \begin{subfigure}[b]{0.24\textwidth}
        \centering
        \includegraphics[width=\textwidth]{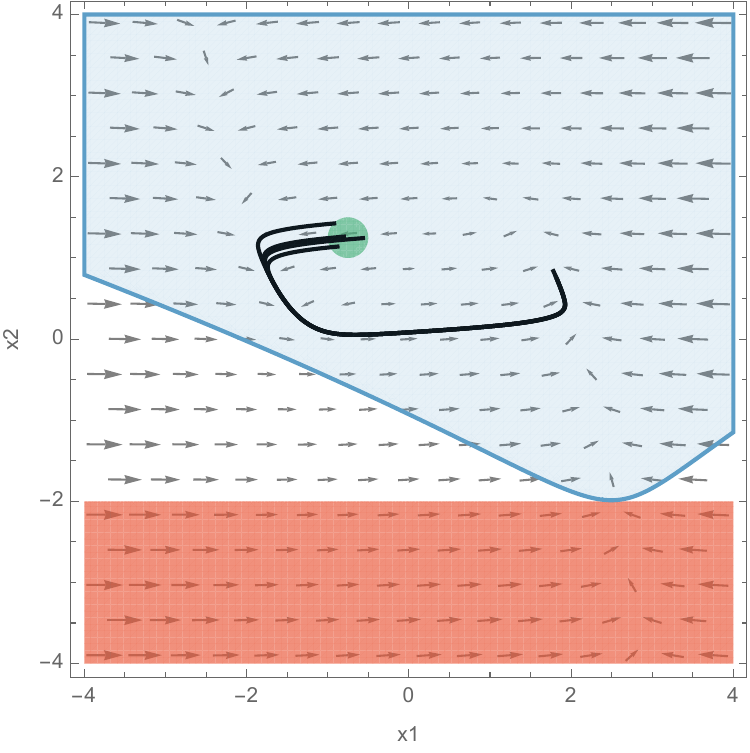}
        \caption{$\textsf{nagumo-2}$}
    \end{subfigure}
    \hfill
    \begin{subfigure}[b]{0.24\textwidth}
        \centering
        \includegraphics[width=\textwidth]{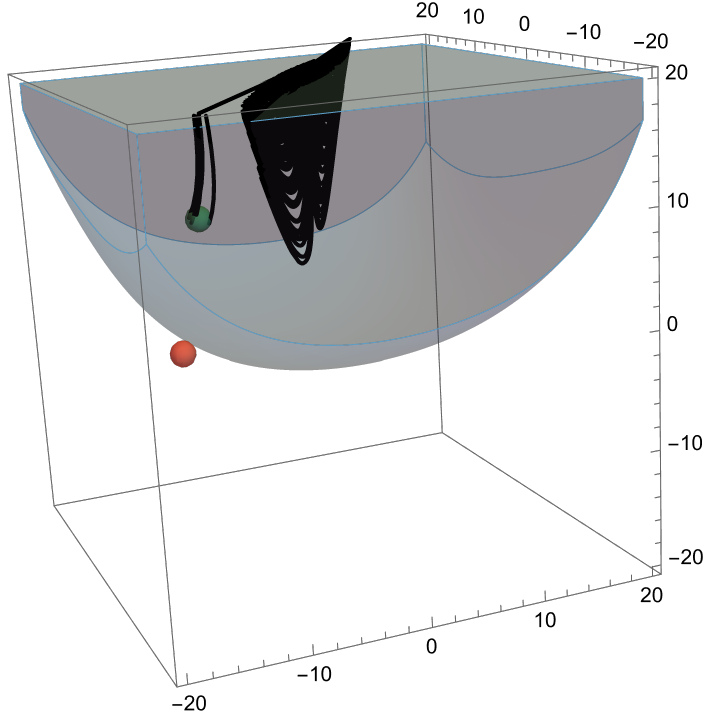}
        \caption{$\textsf{lorenz-1}$}
    \end{subfigure}
    \scriptsize{
    green region: initial set~$\init$; 
    red region: unsafe region~$\unsafe$;
    black solid curves: sampled trajectories $\traj_{\seq x_0}$;\\
    light blue, yellow, and purple region: sub-level set of the synthesized barrier certificate by using Thm.~\ref{thm:bounded}, Thm.~\ref{thm:unbounded-poly}, and Thm.~\ref{thm:unbounded-semi}.
    }
    \caption{Portraits of four selected examples.}
    \label{fig:plot}
\end{figure*}

\section{Conclusion}
\label{sec:summary}

This paper addresses the problem of synthesizing barrier certificates over unbounded domains. 
Previous SDP-based approaches to this problem are incomplete, because Putinar's Positivstellensatz is only applicable in bounded cases.
We fill this gap by proposing the first complete sum-of-squares characterization for polynomial barrier certificates, achieved through the utilization of the homogenization approach derived from optimization theory.
Furthermore, we introduce the notions of homogenized systems and semialgebraic barrier certificates, which are induced from polynomial barrier certificates of the homogenized systems.
For such non-polynomial barrier certificates, we also provide a complete characterization.
Experimental results substantiate the efficacy of both of our approaches, demonstrating their enhanced expressiveness and ability to synthesize more barrier certificates in comparison to existing methods.

While our paper primarily focuses on synthesizing barrier certificates for differential dynamical systems, it is crucial to note that our method can be readily extended to other types of systems, including hybrid systems and systems with control, disturbance, or stochastic dynamics. 
Furthermore, our method can also be utilized in related verification problems such as Lyapunov function synthesis, program invariant generation, and so on.

\begin{credits}
\subsubsection{\ackname} 
This work has been partially funded by the National Key R\&D Program of China under grant No.\ 2022YFA1005101 and 2022YFA1005102, by the NSFC under grant No.\ 62192732 and 62032024, by the CAS Project for Young Scientists in Basic Research under grant No.\ YSBR-040, by the Strategic Priority Research Program of the Chinese Academy of Sciences XDB0640000 \& XDB0640200, by the Key R\&D Program of Hubei Province (2023BAB170), and by the Fundamental Research Funds for the Central Universities.
\end{credits}

%\clearpage
\bibliographystyle{splncs04}
\bibliography{bibfiles/main.bib}

\end{document}